\documentclass[reqno]{amsart}

\usepackage{amssymb}
\usepackage{enumitem}
\usepackage{tikz}
\usepackage{hyperref}

\newcommand*{\mailto}[1]{\href{mailto:#1}{\nolinkurl{#1}}}

\newtheorem{theorem}{Theorem}[section]
\newtheorem{definition}[theorem]{Definition}

\newtheorem{proposition}[theorem]{Proposition}
\newtheorem{corollary}[theorem]{Corollary}

\newcommand{\R}{{\mathbb R}}
\newcommand{\N}{{\mathbb N}}

\newcommand{\C}{{\mathbb C}}

\newcommand{\E}{\mathrm{e}}

\newcommand{\sgn}{\mathrm{sgn}}

\newcommand{\im}{\mathrm{Im}}

\newcommand{\oo}{o}

\numberwithin{equation}{section}

\begin{document}

\title[A coupling problem for entire functions]{A coupling problem for entire functions and its application to the long-time asymptotics of integrable wave equations}

\author[J.\ Eckhardt]{Jonathan Eckhardt}
\address{Faculty of Mathematics\\ University of Vienna\\ Oskar-Morgenstern-Platz 1\\ 1090 Wien\\ Austria}
\email{\mailto{jonathan.eckhardt@univie.ac.at}}
\urladdr{\url{http://homepage.univie.ac.at/jonathan.eckhardt/}}

\author[G.\ Teschl]{Gerald Teschl}
\address{Faculty of Mathematics\\ University of Vienna\\
Oskar-Morgenstern-Platz 1\\ 1090 Wien\\ Austria\\ and International
Erwin Schr\"odinger
Institute for Mathematical Physics\\ Boltzmanngasse 9\\ 1090 Wien\\ Austria}
\email{\mailto{Gerald.Teschl@univie.ac.at}}
\urladdr{\url{http://www.mat.univie.ac.at/~gerald/}}

\thanks{Nonlinearity {\bf 29}, 1036--1046 (2016)}
\thanks{{\it Research supported by the Austrian Science Fund (FWF) under Grant No.\ Y330 and by the AXA Research Fund under the Mittag-Leffler Fellowship Project}}

\keywords{Coupling problem, long-time asymptotics, Camassa--Holm equation}
\subjclass[2010]{Primary 37K40, 35Q35; Secondary 30D20, 37K20}

\begin{abstract}
We propose a novel technique for analyzing the long-time asymptotics of integrable wave equations in the case when the
underlying isospectral problem has purely discrete spectrum. To this end, we introduce a natural coupling problem for entire functions,
which serves as a replacement for the usual Riemann--Hilbert problem, which does not apply in these cases. 
As a prototypical example, we investigate the long-time asymptotics of the dispersionless Camassa--Holm equation.
\end{abstract}

\maketitle

\section{Introduction}

Integrable wave equations play a key role in understanding numerous phenomena in science. In this connection, understanding
the long-time asymptotics of solutions is crucial. Roughly speaking, the typical behavior is that any (sufficiently fast) decaying
initial profile splits into a number of solitons plus a decaying dispersive part. This has been first observed numerically for the
Korteweg--de Vries equation \cite{zakr}. Corresponding asymptotic formulas were derived and justified with increasing level of rigor
over the last thirty years. To date, the most powerful method for deriving such long-time asymptotics is the nonlinear steepest
descent method from Deift and Zhou \cite{dz}, which was inspired by earlier work of Manakov \cite{ma} and Its \cite{its1}.
More on this method and its history can be found in the survey \cite{diz}; an expository introduction to this method
for the Korteweg--de Vries equation can be found in \cite{grte}.

Although this method has found to be applicable to a wide
range of integrable wave equations, there are still some exceptions. The most notable one is the Camassa--Holm 
equation, also known as the dispersive shallow water equation,
\begin{equation}\label{ch.eq}
u_{t} -u_{xxt}  +2\varkappa u_{x} = 2u_x u_{xx} - 3uu_x +u u_{xxx},\quad  x,\, t\in\R,
\end{equation}
where $u\equiv u(x,t)$ is the fluid velocity in the $x$ direction, $\varkappa\geq 0$ is a constant related to the critical shallow water
wave speed, and subscripts denote partial derivatives. It was first introduced by Camassa and Holm in \cite{ch}
and Camassa et al.\ \cite{chh} as a model for shallow water waves, but it actually already appeared earlier in a list by Fuchssteiner and Fokas \cite{ff}. Regarding the hydrodynamical relevance of this equation, let us also mention the more recent articles by Johnson \cite{jo}, Ionescu-Kruse \cite{io07} as well as Constantin and Lannes \cite{cla}.

While in the case $\varkappa>0$ there is an underlying Riemann--Hilbert problem which can be analyzed using the
nonlinear steepest descent method \cite{bosh}, \cite{bosh2}, \cite{boitsh}, \cite{bokoshte} (cf.\ also \cite{cogeiv06} where a related additive Riemann--Hilbert problem is mentioned),
this breaks down in the limiting case $\varkappa=0$. In this case, the solitons are no longer smooth but have a single peak and hence are also known as peakons. 
Nevertheless, it was conjectured by McKean \cite{mckean} (cf.\ also \cite{mckean2}, \cite{mckean3}) that solutions split into a train of peakons, in accordance with earlier numerical observations by Camassa et al.\ \cite{chh}. 
However, apart from the multi-peakon case \cite{bss} (and some low-regularity solutions \cite{li} as well as for a simplified flow \cite{lou}), this has been an open problem,  resolved only recently by us in \cite{IsospecCH}. The technical problem here stems from the fact that the underlying isospectral problem has
purely discrete spectrum and hence it is no longer possible to set up the usual scattering theory. Our approach in \cite{IsospecCH}
circumvented this shortcoming by a thorough investigation of the associated isospectral problem, which then allowed us to deduce 
long-time asymptotics. However, this approach has still some drawbacks. 
For example, it is not possible to obtain long-time asymptotics which hold uniformly in sectors. 

The aim of the present article is to propose a novel approach to such kind of problems, which seems to be more natural. 
In some sense, it can be thought of as an adaptation of the usual Riemann--Hilbert problem approach. 
More precisely, we will replace the Riemann--Hilbert problem with a certain coupling problem for entire functions. 
Consequently, we will investigate the asymptotic behavior of solutions to this problem under known behavior of the given data.  

As a prototypical example, we will apply our results to derive long-time asymptotics for the dispersionless Camassa--Holm equation. 
However, we expect that this new technique will also work for other equations, whose underlying isospectral problem exhibits purely discrete spectrum.
 For example, it can immediately be applied to derive long-time asymptotics for corresponding equations in the whole Camassa--Holm hierarchy. 
 While for the positive members of this hierarchy one gets qualitatively the same asymptotic picture, the situation is somewhat different for the negative ones (including for instance the extended Dym equation). 
 Although solutions of negative members of the Camassa--Holm hierarchy still split into a train of peakons, their speed will be proportional to the modulus of the corresponding eigenvalue. 
 This causes the larger peakons to be the slower ones and the smaller peakons to be the faster ones, creating a qualitatively different picture.

% % % % % % % % % % % %
\section{Coupling problem}
% % % % % % % % % % % %

The purpose of this section is to introduce the notion of a {\em coupling problem} for entire functions. 
To this end, consider a fixed discrete set $\sigma\subseteq\R$ such that the sum
\begin{align}\label{eqnTraceClass}
 \sum_{\lambda\in\sigma} \frac{1}{|\lambda|} 
\end{align}
is finite. 
It is well known that under this condition, the infinite product 
\begin{align}\label{eqnDefWronski}
 W(z) = \prod_{\lambda\in\sigma} \biggl( 1-\frac{z}{\lambda} \biggr), \quad z\in\C,
\end{align}
converges locally uniformly to an entire function of exponential type zero \cite[Lemma 2.10.13]{bo54}, \cite[Theorem~5.3]{le96}. 
Furthermore, we introduce the quantities $\eta_\lambda\in\R\cup\lbrace\infty\rbrace$ for each $\lambda\in\sigma$ which are referred to as the {\em coupling constants}. 

\begin{definition}
 A solution of the coupling problem with data $\lbrace \eta_\lambda \rbrace_{\lambda\in\sigma}$ is a pair of entire functions $(\Phi_-,\Phi_+)$ of exponential type zero such that the following three conditions are satisfied: 
\begin{enumerate}[label=\emph{(\roman*)}, leftmargin=*, widest=W]
\item[{\em (C)}] Coupling condition:
\begin{align*}
 \Phi_+(\lambda) = \eta_\lambda \Phi_-(\lambda), \quad \lambda\in\sigma
\end{align*}
\item[{\em (G)}] Growth and positivity condition: 
\begin{align*}
  \im \left( \frac{z\, \Phi_-(z) \Phi_+(z)}{W(z)}\right) \geq 0, \quad \im(z) > 0
\end{align*}
\item[{\em (N)}] Normalization condition: 
\begin{align*}
 \Phi_-(0) = \Phi_+(0) = 1
\end{align*}
\end{enumerate}
\end{definition}

In order to be precise, if $\eta_\lambda = \infty$ for some $\lambda\in\sigma$, then the coupling condition (C) in this definition has to be read as $\Phi_-(\lambda)=0$. 
The growth and positivity condition (G) means that the meromorphic function
\begin{align}\label{eqnGHN}
 \frac{z\, \Phi_-(z) \Phi_+(z)}{W(z)}, \quad z\in\C\backslash\R,
\end{align}
is a so-called Herglotz--Nevanlinna function, which satisfy certain growth restrictions (to be seen from their integral representations; \cite[Chapter~6]{akgl93}, \cite[Chapter~5]{roro94}). 
Moreover, let us mention that since the residues of such a function are known to be nonpositive, condition (G) also requires the necessary presumption  
\begin{align}\label{eqnHNRes}
   \frac{\lambda\, \Phi_-(\lambda)^2}{\dot{W}(\lambda)}  \eta_\lambda \leq 0, \quad \lambda\in\sigma, 
\end{align}
on the sign of all coupling constants which are finite. 
Thus the coupling constants corresponding to the smallest (in modulus) positive and negative element of $\sigma$ have to be nonnegative. 
The consecutive coupling constants then have to be alternating nonpositive and nonnegative. 
Furthermore, the condition (G) also tells us that the zeros of the numerator and the denominator of the function in~\eqref{eqnGHN} are interlacing (after cancelation) \cite[Theorem~27.2.1]{le96}. 
In particular, this guarantees that the sums  
\begin{align}
 \sum_{\mu\in\rho_\pm} \frac{1}{|\mu|}
\end{align}
are finite, where $\rho_\pm$ denote the sets of all (necessarily simple) zeros of the functions $\Phi_\pm$.  
As a consequence, these functions can be written as the canonical products 
\begin{align}\label{eqnPhiProd}
 \Phi_\pm(z) = \prod_{\mu\in\rho_\pm} \biggl( 1 - \frac{z}{\mu} \biggr), \quad z\in\C,
\end{align}
bearing in mind the normalization condition (N). 
Finally, we mention the bounds 
\begin{align}\label{eqnPhiBound}
 |\Phi_\pm(z)| \leq \prod_{\lambda\in\sigma} \biggl( 1 + \frac{|z|}{|\lambda|} \biggr), \quad z\in\C,
\end{align}
upon roughly estimating~\eqref{eqnPhiProd} and employing the interlacing condition once more. 

Obtaining existence and uniqueness results for the coupling problem is an intricate task which is essentially equivalent to solving the inverse problem for the isospectral problem of the Camassa--Holm equation. 
However, in the simple case when the set $\sigma$ consists of only one point, it is indeed possible to write down the solution explicitly in terms of the one single coupling constant. 

\begin{proposition}\label{propOneSol}
 Suppose that $\sigma = \lbrace \lambda_0 \rbrace$ for some nonzero $\lambda_0\in\R$. 
 If the coupling constant $\eta_{\lambda_0}\in\R\cup\lbrace\infty\rbrace$ is not negative, then the coupling problem has a unique solution given by
 \begin{align}
  \Phi_\pm(z) = \left( 1 - z \frac{1- \min\left(1,\eta_{\lambda_0}^{\pm 1}\right)}{\lambda_0} \right), \quad z\in\C.
 \end{align}
\end{proposition}

\begin{proof}
 It is readily verified that the given polynomials are indeed a solution of the coupling problem.
 Conversely, if $(\Psi_-,\Psi_+)$ is another solution, then 
 \begin{align*}
  \Psi_\pm(z) = 1 - a_\pm z, \quad z\in\C,
 \end{align*}  
 for some $a_\pm\in\R$ with $a_- a_+ = 0$. 
 Moreover, we infer that $0\leq a_\pm \lambda_0 \leq 1$ in view of the Herglotz--Nevanlinna property (more precisely, from the interlacing condition of the poles and zeros).
 Lastly, the coupling condition (C) takes the form   
 \begin{align*}
  1 - a_+ \lambda_0 = \eta_{\lambda_0} \left( 1 - a_- \lambda_0 \right). 
 \end{align*}
 Now if $\eta_{\lambda_0}\leq 1$, then necessarily $a_-=0$ since otherwise we get the contradiction
 \begin{align*}
  1 = \eta_{\lambda_0} (1-a_- \lambda_0) < 1.
 \end{align*} 
 Consequently, we may express $a_+$ in terms of the coupling constant using the coupling condition. 
 In much the same manner, one may obtain the coefficients $a_\pm$ if $\eta_{\lambda_0}\geq 1$ and finally end up with  
 \begin{align*}
  a_\pm = \frac{1- \min\left(1,\eta_{\lambda_0}^{\pm 1}\right)}{\lambda_0}
 \end{align*}
 in either case, which finishes the proof. 
\end{proof}

Note that there is no solution of the coupling problem in Proposition~\ref{propOneSol} if the coupling constant is negative, since it would violate the positivity condition (G).

% % % % % % % % % % % %
\section{Asymptotic analysis}
% % % % % % % % % % % %

We shall now derive a general result on the asymptotic behavior of solutions to the coupling problem.
Therefore, let $\Delta$ be a first-countable topological space (that is, every point has a countable neighborhood basis) and fix some $\delta_\infty\in\Delta$.
Again, we denote with $\sigma\subseteq\R$  a discrete set such that the sum~\eqref{eqnTraceClass} is finite and define the entire function $W$ by~\eqref{eqnDefWronski}. 
Moreover, for every $\delta\in\Delta$ we consider a set of coupling constants $\eta_\lambda(\delta)\in\R\cup\lbrace\infty\rbrace$ indexed by $\lambda\in\sigma$. 

\begin{theorem}\label{thmAA}
 Suppose there is a partition $\sigma_- \cup \lbrace\lambda_0\rbrace \cup \sigma_+$ of $\sigma$  such that   
 \begin{align}\label{eqnConvCC}
  -\ln |\eta_{\lambda}(\delta)| \rightarrow  \pm\infty 
 \end{align}
 as $\delta\rightarrow\delta_\infty$ for each $\lambda\in\sigma_\pm$ and define the conjugated coupling constants by 
 \begin{align}
  \hat{\eta}_{\lambda_0}(\delta) = \eta_{\lambda_0}(\delta) \prod_{\lambda\in\sigma_-}\biggl(1-\frac{\lambda_0}{\lambda}\biggr) \prod_{\lambda\in\sigma_+}\biggl(1-\frac{\lambda_0}{\lambda}\biggr)^{-1}, \quad \delta\in\Delta. 
 \end{align}
 If the pairs $(\Phi_-(\,\cdot\,,\delta),\Phi_+(\,\cdot\,,\delta))$ are solutions of the coupling problem with data $\lbrace \eta_{\lambda}(\delta)\rbrace_{\lambda\in\sigma}$  for every $\delta\in \Delta$, then 
 \begin{align}\label{eqnConv}
  \frac{\Phi_-(z,\delta)\Phi_+(z,\delta)}{W(z)} = 1 + \frac{z}{\lambda_0- z}  \min\left(\hat{\eta}_{\lambda_0}(\delta)^{-1},\hat{\eta}_{\lambda_0}(\delta)\right) + \oo(1)
 \end{align}
  as $\delta\rightarrow\delta_\infty$, locally uniformly for all $z\in\C\backslash\sigma$. 
\end{theorem}

\begin{proof}
First consider a sequence $\delta_k\in \Delta$, $k\in\N$ with $\delta_k\rightarrow\delta_\infty$ as $k\rightarrow\infty$ such that the entire functions $\Phi_\pm(\,\cdot\,,\delta_k)$ converge locally uniformly as $k\rightarrow\infty$.
The respective limits are entire functions of exponential type zero in view of~\eqref{eqnPhiBound} and will be denoted by $\Psi_\pm$. 
Due to assumption~\eqref{eqnConvCC} and the coupling condition, we conclude that $\Psi_\pm(\lambda) = 0$ for $\lambda\in\sigma_\pm$ (also observe that the quantities $\Phi_\pm(\lambda,\delta)$ are uniformly bounded in $\delta\in\Delta$). 
As a consequence, the meromorphic Herglotz--Nevanlinna function 
\begin{align}\label{eqnPsiHN}
 \frac{z\, \Psi_-(z) \Psi_+(z)}{W(z)}, \quad z\in\C\backslash\R,
\end{align}
 has only one pole and thus at most two zeros, which are necessarily simple. 
 Consequently, we may write (keep in mind that these functions are of exponential type zero and that their zeros have genus zero) 
\begin{align*}
 \Psi_\pm(z) = P_\pm(z) \prod_{\lambda\in\sigma_\pm} \biggl( 1-\frac{z}{\lambda} \biggr), \quad z\in\C,
\end{align*}
where $P_\pm$ are polynomials such that $P_- P_+$ has at most one zero, which is simple. 
Moreover, the pair $(P_-,P_+)$ satisfies the coupling condition
\begin{align*}
 P_+(\lambda_0) = \eta_{\lambda_0,\infty}  P_-(\lambda_0), 
\end{align*}
where the constant $\eta_{\lambda_0,\infty}\in\R\cup\lbrace \infty\rbrace$ is given as the limit
\begin{align*}
 \eta_{\lambda_0,\infty} = \prod_{\lambda\in\sigma_-} \biggl( 1 - \frac{\lambda_0}{\lambda}\biggr) \prod_{\lambda\in\sigma_+} \biggl( 1 - \frac{\lambda_0}{\lambda}\biggr)^{-1} \lim_{k\rightarrow\infty} \Phi_+(\lambda_0,\delta_k) \Phi_-(\lambda_0,\delta_k)^{-1}. 
\end{align*}
Hereby note that the limit is nonnegative because of~\eqref{eqnHNRes}. 
In view of Proposition~\ref{propOneSol} we now may write down the polynomials $P_\pm$ explicitly and conclude that
\begin{align*}
 \frac{\Phi_-(z,\delta_k) \Phi_+(z,\delta_k)}{W(z)} \rightarrow 1 + \frac{z}{\lambda_0 - z} \min\left( \eta_{\lambda_0,\infty}^{-1}, \eta_{\lambda_0,\infty} \right)
\end{align*}
as $k\rightarrow\infty$, locally uniformly in $z\in\C\backslash\sigma$. 
Finally, from the very definition of the constants $\eta_{\lambda_0,\infty}$ we may also rewrite this as  
\begin{align*}
 \frac{\Phi_-(z,\delta_k) \Phi_+(z,\delta_k)}{W(z)} =  1 + \frac{z}{\lambda_0- z}  \min\left(\hat{\eta}_{\lambda_0}(\delta_k)^{-1},\hat{\eta}_{\lambda_0}(\delta_k)\right) + \oo(1)
\end{align*}
as $k\rightarrow\infty$, locally uniformly in $z\in\C\backslash\sigma$. 

Finally, if the claim of the theorem was not true, then there would be a compact set $K\subseteq\C\backslash\sigma$ and a subsequence $\delta_k\in \Delta$, $k\in\N$ with $\delta_k\rightarrow\delta_\infty$ as $k\rightarrow\infty$ such that 
\begin{align}\label{eqnEps}
 \left|   \frac{\Phi_-(z,\delta_k) \Phi_+(z,\delta_k)}{W(z)} - 1 - \frac{z}{\lambda_0- z}  \min\left(\hat{\eta}_{\lambda_0}(\delta_k)^{-1},\hat{\eta}_{\lambda_0}(\delta_k)\right) \right| > \varepsilon  
\end{align}
for all $z\in K$,  $k\in\N$ and some $\varepsilon>0$. 
However, a compactness argument (recall~\eqref{eqnPhiBound} and apply Montel's theorem) shows that there is a subsequence $\delta_{k_l}$ such that $\Phi_\pm(\,\cdot\,,\delta_{k_l})$ converges locally uniformly as $l\rightarrow\infty$. 
In view of the first part of the proof, this gives a contradiction to~\eqref{eqnEps}.
\end{proof}

The assumptions in Theorem~\ref{thmAA} allow one of the coupling constants to be arbitrary.  
This will turn out to be crucial to obtain long-time asymptotics of the Camassa--Holm equation which are valid uniformly in sectors. 
 However, in the case when all of the coupling constants are supposed to converge to zero or infinity, one obtains the following result. 

\begin{corollary}\label{corAA}
 Suppose that we have $|\ln|\eta_\lambda(\delta)|| \rightarrow \infty$ as $\delta\rightarrow\delta_\infty$ for every $\lambda\in\sigma$. 
  If the pairs $(\Phi_-(\,\cdot\,,\delta),\Phi_+(\,\cdot\,,\delta))$ are solutions of the coupling problem with data $\lbrace \eta_{\lambda}(\delta)\rbrace_{\lambda\in\sigma}$  for every $\delta\in \Delta$, then 
 \begin{align}
  \frac{\Phi_-(z,\delta)\Phi_+(z,\delta)}{W(z)} \rightarrow 1
 \end{align}
  as $\delta\rightarrow\delta_\infty$, locally uniformly for all $z\in\C\backslash\sigma$. 
\end{corollary}

\begin{proof}
 Similarly to the first part of the proof of Theorem~\ref{thmAA}, one infers that 
 \begin{align*}
  \frac{\Phi_-(z,\delta_k)\Phi_+(z,\delta_k)}{W(z)} \rightarrow 1 
 \end{align*}
 as $k\rightarrow\infty$, locally uniformly for all $z\in\C\backslash\sigma$ as long as the functions $\Phi_\pm(\,\cdot\,,\delta_k)$ are assumed to converge locally uniformly. 
 In fact, this follows since the function in~\eqref{eqnPsiHN} is now known to have no poles at all. 
 Now the claim follows in much the same manner as in the second part of the proof of Theorem~\ref{thmAA}, invoking a compactness argument. 
\end{proof}

% % % % % % % % % % % % % % 
\section{Applications to the Camassa--Holm equation}
% % % % % % % % % % % % % % 

As anticipated in the introduction, we will now demonstrate that our results provide a powerful tool to derive long-time asymptotics for the dispersionless Camassa--Holm equation. 
To this end, let $u$ be a solution of  
\begin{align}
  u_t - u_{xxt} = u u_{xxx} - 3 u u_x + 2 u_x u_{xx} 
\end{align}
with decaying spatial asymptotics. 
To be precise, we will assume that the quantities
\begin{align}
 \omega(x,t) = u(x,t) - u_{xx}(x,t), \quad x\in\R,
\end{align}
are finite signed measure for each time $t\geq 0$. 

These conditions guarantee (see \cite[Theorem~3.1]{IsospecCH}) that for every time $t\geq 0$ and  $z\in\C$, there are unique solutions $\phi_\pm(z,\cdot\,,t)$ of the differential equation 
\begin{align}\label{eqnISO}
 -\phi_\pm''(z,x,t) + \frac{1}{4} \phi_\pm(z,x,t) = z\, \omega(x,t) \phi_\pm(z,x,t), \quad x\in\R,
\end{align}
(the prime denotes spatial differentiation) with the spatial asymptotics
\begin{align}
 \phi_\pm(z,x,t) & \sim \E^{\mp\frac{x}{2}}, & \phi_\pm'(z,x,t) & \sim \mp\frac{1}{2} \E^{\mp\frac{x}{2}},
\end{align}
as $x\rightarrow\pm\infty$. 
In view of \cite[Theorem~4.1]{IsospecCH}, it is known that these solutions are real entire and of exponential type zero with respect to the spectral parameter. 

 Now the importance of the spectral problems~\eqref{eqnISO} lies in the well-known fact that their spectra are invariant under the Camassa--Holm flow, that is, they are the same for all times $t\geq0$ (for example, we refer to \cite[Section~2]{bss}, \cite{ch}, \cite[Section~3]{co01}, \cite[Theorem~5.1]{ConservMP}). 
 For this reason, we may simply denote the spectrum of~\eqref{eqnISO} with $\sigma$, which is known to be real and purely discrete such that the sum   
 \begin{align}
  \sum_{\lambda\in\sigma} \frac{1}{|\lambda|}  
 \end{align}
 is finite, in view of \cite[Proposition~3.3]{IsospecCH}.  
 The Wronskian of our two solutions  
 \begin{align}
  W(z) & = \phi_+(z,x,t) \phi_-'(z,x,t) - \phi_+'(z,x,t) \phi_-(z,x,t), \quad z\in\C,
 \end{align}
 turns out to be independent of space $x\in\R$ and time $t\geq0$. 
 Indeed, this function is the characteristic function of the spectral problem~\eqref{eqnISO}, that is,  
 \begin{align}
      W(z)  = \prod_{\lambda\in\sigma} \biggl( 1-\frac{z}{\lambda} \biggr), \quad z\in\C, 
 \end{align}
 in view of \cite[Corollary~4.2]{IsospecCH}. 

In order to point out the connection to the coupling problem for entire functions, one observes that the solutions $\phi_+(\lambda,\cdot\,,t)$ and $\phi_-(\lambda,\cdot\,,t)$ are linearly dependent for every eigenvalue $\lambda\in\sigma$ and time $t\geq 0$. Hence 
 there is some nonzero real $c_\lambda(t)\in\R$ such that we may write 
\begin{align}
  \phi_+(\lambda,x,t) = c_\lambda(t) \phi_-(\lambda,x,t), \quad x\in\R. 
\end{align}
The time evolution for these quantities is known to be given explicitly by 
\begin{align}
 c_\lambda(t) = c_{\lambda}(0) \E^{-\frac{t}{2\lambda}}, \quad t\geq 0,~\lambda\in\sigma. 
\end{align}
More precisely, this follows from the well-known time evolution of the associated norming constants (for example, see \cite[Section~2]{bss}, \cite[Section~3]{co01}, \cite[Theorem~5.1]{ConservMP}) and the identity in \cite[Lemma~3.2]{IsospecCH},
\begin{align}
 - \dot{W}(\lambda) = c_\lambda(t) \int_\R \phi_-(\lambda,x,t)^2 d\omega(x,t), \quad t\geq0,~\lambda\in\sigma,
\end{align}
where the dot denotes differentiation with respect to the spectral parameter. 

We have now collected all necessary ingredients to prove the announced long-time asymptotics for the solution $u$ of the Camassa--Holm equation. 
In fact, the proof of this result is almost immediate from the general results on asymptotic analysis for our coupling problem of entire functions derived in the previous section. 

\begin{theorem}\label{thmLT}
 Let $S\subseteq\R\times(0,\infty)$ be a closed sector which contains at most finitely many of the rays $\mathrm{r}_\lambda$, $\lambda\in\sigma$ given by $2\lambda x = t$.
  Then we have 
 \begin{align}\label{eqnLTu}
  u(x,t) = \sum_{\lambda\in\sigma} \frac{1}{2\lambda} \E^{-\left|x-\frac{t}{2\lambda} + \xi_\lambda\right|} + \oo(1)
 \end{align}
 for $(x,t)\in S$ as $t\rightarrow\infty$,  where the phase shifts $\xi_\lambda$ for each $\lambda\in\sigma$ are given by
 \begin{align}
   \xi_\lambda = \ln|c_{\lambda}(0)| + \sum_{\kappa\in\sigma\backslash\lbrace\lambda\rbrace} \sgn\left( \frac{1}{\lambda} - \frac{1}{\kappa} \right) \ln\left|1-\frac{\lambda}{\kappa}\right|, \quad \lambda\in\sigma.
 \end{align}
\end{theorem}

\begin{proof}
 For every $x\in\R$ and $t\geq 0$ we introduce the entire functions
 \begin{align*}
  \Phi_\pm(z,x,t) = \E^{\pm\frac{x}{2}} \phi_\pm(z,x,t), \quad z\in\C, 
 \end{align*}
 which will turn out to be a solution of a particular coupling problem. 
 In fact, one clearly has the coupling condition (C) 
 \begin{align}\label{eqnLTcc}
  \Phi_+(\lambda,x,t) = \E^{x-\frac{t}{2\lambda}} c_\lambda(0) \Phi_-(\lambda,x,t), \quad \lambda\in\sigma. 
 \end{align}
 Moreover, due to \cite[Proposition~4.4]{LeftDefiniteSL}, the function 
 \begin{align*}
  \frac{z\, \Phi_-(z,x,t) \Phi_+(z,x,t)}{W(z)} = \left( \frac{\phi_-'(z,x,t)}{z\phi_-(z,x,t)} - \frac{\phi_+'(z,x,t)}{z\phi_+(z,x,t)} \right)^{-1}, \quad z\in\C\backslash\R,
 \end{align*}
 is a Herglotz--Nevanlinna function, ensuring the growth and positive condition (G). 
 In fact, this can also be verified by a direct calculation, using the differential equation~\eqref{eqnISO}. 
 Finally, the normalization condition (N) is immediate from the definition. 
 
 We will first consider the special case when the sector $S$ contains precisely one ray, say $\mathrm{r}_{\lambda_0}$ for some $\lambda_0\in\sigma$. 
 Upon defining the sets $\sigma_\pm\subseteq\sigma$ by
 \begin{align*}
   \sigma_\pm = \left\lbrace \lambda\in\sigma \,|\, \pm \lambda_0^{-1} < \pm \lambda^{-1} \right\rbrace,
 \end{align*}
 one obtains a partition $\sigma_-\cup\lbrace\lambda_0\rbrace\cup\sigma_+$ of $\sigma$ such that
 \begin{align*}
  \mp \biggl(\frac{x}{t} - \frac{1}{2\lambda}\biggr) > \varepsilon_\lambda, \quad (x,t)\in S,~\lambda\in\sigma_\pm, 
 \end{align*} 
 for some $\varepsilon_\lambda>0$, $\lambda\in\sigma_\pm$. 
 Therefore, the coupling constants in~\eqref{eqnLTcc} satisfy
 \begin{align*}
   -\ln \left|\E^{x-\frac{t}{2\lambda}} c_\lambda(0)\right| = -\left(\frac{x}{t}-\frac{1}{2\lambda}\right) t - \ln|c_\lambda(0)| \rightarrow \pm\infty, \quad \lambda\in\sigma_\pm,
 \end{align*} 
 for $(x,t)\in S$ as $t\rightarrow\infty$. 
 In view of \cite[Lemma~3.4]{IsospecCH} and Theorem~\ref{thmAA} this yields
 \begin{align*}
  u(x,t) = \frac{1}{2} \left. \frac{\partial}{\partial z} \frac{\Phi_-(z,x,t) \Phi_+(z,x,t)}{W(z)}\right|_{z=0} = \frac{1}{2\lambda_0} \E^{-\left| x - \frac{t}{2\lambda_0} + \xi_{\lambda_0} \right|}  + \oo(1)
 \end{align*}
 for $(x,t)\in S$ as $t\rightarrow\infty$. 
 But this proves the claim in this special case, since 
 \begin{align*}
  \sum_{\lambda\in\sigma\backslash\lbrace\lambda_0\rbrace} \frac{1}{2\lambda} \E^{-\left|x-\frac{t}{2\lambda} + \xi_\lambda\right|} = \oo(1)
 \end{align*}
 for $(x,t)\in S$ as $t\rightarrow\infty$, as Lebesgue's dominated convergence theorem shows. 
 
 In order to finish the proof in the general case, notice that under our assumptions we may cover the sector $S$ with finitely many sectors of the type considered above. 
\end{proof}

 The typical long-time behavior of a solution $u$ of the Camassa--Holm equation, derived in Theorem~\ref{thmLT} can be depicted as follows:
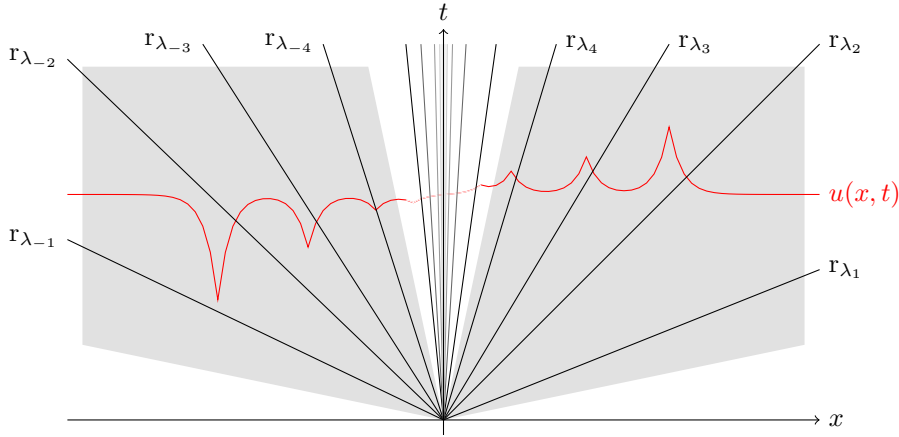
\begin{center}
 \begin{tikzpicture}[domain=-1:9, samples=101]

\fill[black!12] (4,0) --  (8.8, 1) -- (8.8,4.7) -- (5,4.7) -- cycle;
\fill[black!12] (4,0) --  (-0.8, 1) -- (-0.8,4.7) -- (3,4.7) -- cycle;
\draw[color=red] plot (\x,{ 3 + 0.9*exp(-6*abs(\x-7.)) + 0.5*exp(-6*abs(\x-5.9)) + 0.3*exp(-6*abs(\x-4.9)) - 1.4*exp(-6*abs(\x-1)) - 0.7*exp(-6*abs(\x-2.2)) - 0.2*exp(-6*abs(\x-3.1))+ 0.1*exp(-6*abs(\x-4.5)) - 0.1*exp(-6*abs(\x-3.6)) }) node[right] {$u(x,t)$};
\draw[color=white,domain=3.5:4.5] plot (\x,{ 3 + 0.9*exp(-6*abs(\x-7.)) + 0.5*exp(-6*abs(\x-5.9)) + 0.3*exp(-6*abs(\x-4.9)) - 1.4*exp(-6*abs(\x-1)) - 0.7*exp(-6*abs(\x-2.2)) - 0.2*exp(-6*abs(\x-3.1))+ 0.1*exp(-6*abs(\x-4.5)) - 0.1*exp(-6*abs(\x-3.6)) }) node[right] {};

\draw[-] (4,0) -- (9,2) node[right] {$\mathrm{r}_{\lambda_1}$};
\draw[-] (4,0) -- (9,5) node[right] {$\mathrm{r}_{\lambda_2}$};
\draw[-] (4,0) -- (7,5) node[right] {$\mathrm{r}_{\lambda_3}$};
\draw[-] (4,0) -- (5.5,5) node[right] {$\mathrm{r}_{\lambda_4}$};
\draw[-] (4,0) -- (4.7,5) node[right] {};
\draw[-, color=black!60] (4,0) -- (4.3,5) node[right] {};
\draw[-, color=black!40] (4,0) -- (4.12,5) node[right] {};
\draw[-, color=black!20] (4,0) -- (4.05,5) node[right] {};
\draw[-, color=black!10] (4,0) -- (4.02,5) node[right] {};

\draw[-] (4,0) -- (-1,2.4) node[left] {$\mathrm{r}_{\lambda_{-1}}$};
\draw[-] (4,0) -- (-1,4.8) node[left] {$\mathrm{r}_{\lambda_{-2}}$};
\draw[-] (4,0) -- (0.8,5) node[left] {$\mathrm{r}_{\lambda_{-3}}$};
\draw[-] (4,0) -- (2.4,5) node[left] {$\mathrm{r}_{\lambda_{-4}}$};
\draw[-] (4,0) -- (3.5,5) node[left] {};
\draw[-, color=black!60] (4,0) -- (3.7,5) node[left] {};
\draw[-, color=black!40] (4,0) -- (3.88,5) node[left] {};
\draw[-, color=black!20] (4,0) -- (3.95,5) node[left] {};
\draw[-, color=black!10] (4,0) -- (3.98,5) node[left] {};

\draw[->] (-1,0) -- (9,0) node[right] {$x$};
\draw[->] (4,-0.2) -- (4,5.2) node[above] {$t$};
\end{tikzpicture}
\end{center}
Hereby, the grey areas represent two sectors in which our long-time asymptotics hold uniformly. 
Each of the rays $\mathrm{r}_\lambda$, accumulating at the $t$-axis, corresponds to an eigenvalue $\lambda\in\sigma$ of the underlying isospectral problem. 
 After long time, one can see that the solution $u$ splits into a train of single peakons, each of which travels along one of the rays, with height and speed determined by the corresponding eigenvalue. 

Due to the conditions on the sector in Theorem~\ref{thmLT}, we do not obtain long-time asymptotics of solutions, which hold uniformly in sectors around the $t$-axis (as long as the spectrum is not finite, that is, in the multi-peakon case). 
 However, we are able to derive long-time asymptotics which hold uniformly in strips near the $t$-axis, that is, as long as $x$ stays bounded.  

\begin{corollary}\label{corLTt}
 Given some $R>0$, one has 
 \begin{align}\label{eqnLTfixedx}
  u(x,t) = \oo(1) 
 \end{align}
 for $|x|\leq R$ as $t\rightarrow\infty$. 
\end{corollary}

\begin{proof}
 With the notation from the proof of Theorem~\ref{thmLT} we see that the coupling constants in~\eqref{eqnLTcc} for every $\lambda\in\sigma$ satisfy  
 \begin{align*}
    \left| \ln \left|\E^{x-\frac{t}{2\lambda}} c_\lambda(0)\right| \right| = \left| x-\frac{t}{2\lambda} + \ln|c_\lambda(0)|  \right| \rightarrow \infty
 \end{align*}
 as $t\rightarrow\infty$. 
 Therefore, an application of Corollary~\ref{corAA} shows that 
 \begin{align*}
  u(x,t) = \oo(1)
 \end{align*}
 as $t\rightarrow\infty$, in view of \cite[Lemma~3.4]{IsospecCH}. 
\end{proof}

Note that \eqref{eqnLTfixedx} is consistent with \eqref{eqnLTu} since
\[
\sum_{\lambda\in\sigma} \frac{1}{2\lambda} \E^{-\left|x-\frac{t}{2\lambda} + \xi_\lambda\right|} = \oo(1) 
\]
for $|x|\leq R$ as $t\rightarrow\infty$ by virtue of Lebesgue's dominated convergence theorem.

The fact that the limit of $u(x,t)$ vanishes for every fixed $x\in\R$ as $t\rightarrow\infty$ was established in \cite[Theorem~1.2]{xizh00}, \cite[Theorem~3]{xizh02} for certain weak solutions of the Camassa--Holm equation under various additional assumptions.

\bigskip
\noindent
{\bf Acknowledgments.}
We gratefully acknowledge the kind hospitality of the {\em Institut Mittag-Leffler} (Djursholm, Sweden) during the scientific program on {\em Inverse Problems and Applications} in spring 2013, where most of this article was written.

\end{document}